\documentclass[12pt]{article}
\usepackage{amsmath}
\usepackage{amssymb,amsthm}
\usepackage{color}

\newtheorem{Thm}{Theorem}[section]
\newtheorem{theorem}[Thm]{Theorem}
\newtheorem{proposition}[Thm]{Proposition}
\newtheorem{corollary}[Thm]{Corollary}

\newtheorem{remark}{Remark}[section]

\title{One or two small points in thermodynamics}

\author{Walter F. Wreszinski\footnote{wreszins@gmail.com, 
Instituto de Fisica, Universidade de S\~ao Paulo (USP), Brazil}}        

\begin{document}

\maketitle

\begin{abstract}
I present my recollections of what I used to find to be ``one or two small points in thermodynamics'', following Sommerfeld's famous quote, and review them on the light of present knowledge. 
\end{abstract}

\section{Introduction and motivation}

During my student days, I enjoyed the famous quote from Arnold Sommerfeld \cite{Somm}: `` Thermodynamics is a funny subject. The first time you go through it, you don't understand it at all. The second time you go through it, you think you understand it, except for one or two small points. The third time you go through it, you know you don't understand it, but by that time you are so used to it, that it doesn't bother you any more''.

In fact, almost every student has found his ``one or two small points''. This may be due to the fact that, in spite of being ``the unique universal science'' \cite{Had}, thermodynamics is populated with various logical and mathematical contradictions, to the point that - in another famous quote - the russian mathematician Vladimir I. Arnold was led to assert that ``every mathematician knows it is impossible to understand an elementary course in thermodynamics'' \cite{Arn}.

Having written one such course (textbook) myself, I realized that I succeeded only partially in trying to contradict Arnold's quote, because, even after the completion of the second edition \cite{WreEd}, at least ``two small points'' did remain. The present note is my account of them. The whole discussion (restricted to the first and second laws, for the third, see \cite{WreA}) is elementary, i.e., at the level of an elementary course in thermodynamics. Classical references are \cite{Cal}, \cite{tHW}. A recent textbook, with a large palette of interesting applications, including chemistry, is \cite{Oliveira1}. Although applications are very important, my primary aim here is to be as precise in the foundations as possible, remaining in the textbook level. 

In a slightly more general context, one might inquire: why should one try to arrive at a more precise understanding of the laws of nature? One reason is that, even regarding some apparently rather commonplace phenomena, such as darkness at night, the various \emph{stages of understanding} may lead to surprisingly new insights: who would have suspected that, in a deeper, cosmological, sense, this phenomenon is related to the death of stars (\cite{Har}, Chap. 12, p. 249)? In addition, and intimately related to this issue, there is Lesson III of Arthur Wightman in \cite{Jaffe}: Distinguish ``what you know'' from what you ``think you know''. In this same lesson III, it is also reported that ``Arthur insisted: A great physical theory is not mature until it has been put in a precise mathematical form''. Concerning thermodynamics, this endeavor has been achieved by E. Lieb and J. Yngvason in a remarkable analysis \cite{LYPR}, to which I shall have occasion to refer in the forthcoming sections. Two pedagogical accounts of their work are \cite{LYNAMS} and \cite{LYPT}.

Section 2 discusses the first small point, and section 3 a possible solution, as well as the reasons why this solution remains physically objectionable. Section 4 expands on the point left in section 3 and suggests that the interplay with other areas of physics might be necessary for a better understanding. In
section 5, I discuss the second small point, leaving the very brief conclusion to section 6.

\section{A first small point: Clausius' formulation of the second law}

In my student time, Dyson's article ``What is heat?'' \cite{Dyson2} was very popular. It began with the sentence ``Heat is disordered energy''. indeed, heat, denoted by $Q$, provides the balance of energy in the \emph{first law of thermodynamics}, viz.

\begin{equation}
\tilde{d}Q= dU + \tilde{d}W
\label{(1)}
\end{equation}

Above, $\tilde{d}$ denote ``inexact differentials'' and, in \eqref{(1)}, for a given system, $U$ denotes its (internal) energy, and $W$ the amount of work done by the outside world on it. An \emph{equilibrium state} of a \emph{simple system} will be taken to be characterized by a point $X \equiv (U,V,N)$, with $U$ the internal energy, $V$ a work coordinate, e.g., the volume, and $N$ the particle number. The question of the choice of variables in thermodynamics is important: see \cite{ASWight}, see also \cite{Cout} for a clear exposition.

Reference \cite{Dyson2} discusses the example of a bicycle pump: ``Before compression the air atoms are already moving at random in all directions; in other words, this is a disordered system, and its energy is in the form of heat, though we do not feel it, because the air is only at room temperature. Now, if you pump vigorously, compressing the air rapidly, it heats up; the pump becomes hot to the touch. The air has the same disorder it had before, but more energy. By doing work, you have pushed more energy into the air, and the observed production of heat is just the effect of this addition of energy to the pre-existing disorder.''

The previous description seems to confirm the remarks by Lieb and Yngvason (\cite{LYPR}, p.7): ``no one has ever seen heat, nor will it ever be seen, smelled or touched''. Indeed, before pumping, the energy was \emph{also} in the form of heat, ``although we do not feel it''. The \emph{only} aspect distinguishing heat from other forms of energy is the dependence on the way how the processes in which it enters are carried out, which is directly inherited from \eqref{(1)}, in case $\tilde{d}W \ne 0$.

Continuing with the textbook conventions, a \emph{process} is a transformation $C_{12} = X_{1} \to X_{2}$ from a \emph{state} $X_{1} \equiv (U_{1},V_{1},N_{1})$ to a state $X_{2} = (U_{2},V_{2},N_{2})$. I shall be dealing with \emph{equilibrium}, by which it is meant that a process consists of an infinite number of ``infinitesimal'' transformations, i.e., the limiting idealized case in which each intermediate stage deviates only ``infinitesimally'' from equilibrium. This procedure has been avoided in (\cite{LYPR}, p.17) by their definition of adiabatic accessibility. There are processes $X_{1} \to X_{2}$ for which the amount of work done by the outside world in going from $X_{1} \to X_{2}$ is independent of the manner in which this transition is carried out: they are called \emph{adiabatic}. Otherwise, the first law \eqref{(1)} holds instead, with $\tilde{d}Q \ne 0$ and $\tilde{d}W \ne 0$. If, for any process on a given system, $\tilde{Q}=dU$ in \eqref{(1)}, we say that we have a \emph{heat or energy reservoir}, when, instead, $\tilde{d}Q= \tilde{d}W$ in \eqref{(1)}, one speaks of a \emph{work reservoir} - two important idealized systems. We state the second law of thermodynamics in the Kelvin-Planck form (abbreviated second law (KP)):

\emph{Second law (KP)} No process is possible, the sole result of which is a change of energy of a simple system, without changing the work coordinates, and the raising of a weight.

The above is the paraphrasing of the usual Kelvin-Planck formulation of the second law (\cite{Cal},\cite{tHW}, \cite{WreEd}) due to Lieb and Yngvason (\cite{LYPR}, p.49), without using the concept of heat. This shows that the latter concept may be avoided completely when restricting oneself to the second law.

A process $C_{12} = X_{1} \to X_{2}$ is \emph{possible} if it complies with both the first and second laws. It is \emph{reversible} if the process $C_{21} = X_{2} \to X_{1}$ is also possible, otherwise it is \emph{irreversible}. A \emph{cycle} is a process $X \to X$, for some state $X$: it may be reversible or irreversible. Note that the definition of irreversibility does not involve the time as a parameter, i.e., no \emph{dynamics} is attached to it - which should be no surprise within the framework adopted here, which is that of equilibrium thermodynamics.

By further introduction of the Carnot cycle, together with the model of the ideal gas, one is able to prove Carnot's theorem (e.g., \cite{tHW}, 2.2. p.18), leading to the concept of absolute temperature $T$ (\cite{tHW}, p.20), and the formula

\begin{equation}
\frac{Q_{1}}{Q_{2}} = \frac{T_{1}}{T_{2}}
\label{(2)}
\end{equation}
for a reversible cycle, where $Q_{2}$ is the amount of heat absorbed at temperature $T_{2}$, $Q_{1}$ the amount of heat given off at temperature $T_{1}$, and $T_{2} > T_{1}$. Further elaboration (\cite{tHW}, p.24) leads to consider a system traversing a cyclic process $C$, exchanging heat with a series of heat reservoirs at temperatures $T_{1}, T_{2}, \cdots T_{n}$, $Q_{1}, Q_{2}, \cdots Q_{n}$ being the respective algebraic amounts of heat exchanged, positive when absorbed, and negative when given off by the system. Considering, now, the further process $C^{rev}$, consisting of $n$ reversible Carnot processes between each of the $n$ heat reservoirs at temperatures $T_{1}, \cdots, T_{n}$ and a new reservoir at temperature $T_{0}$, so designed that, in the second process, the quantities $Q_{1}, Q_{2}, \cdots Q_{n}$ are returned to the reservoirs at $T_{1}, T_{2}, \cdots T_{n}$, the composite process $C + C^{rev}$ will yield the result that the $n$ reservoirs are left unchanged. Using \eqref{(2)}, we obtain from the second law (KP) (see also \cite{tHW}, p. 24) 

\emph{Clausius' Inequality (1)}

\begin{equation}
Q_{0} = T_{0} \sum_{i=1}^{n} \frac{Q_{i}}{T_{i}} \le 0
\label{(3)}
\end{equation}
with the equality sign in \eqref{(3)} holding if and only if the process $C$ is also reversible. Taking the limit $n \to \infty$, whereby any cycle $Cy$ may be approached by a mesh of Carnot cycles, we obtain

\begin{equation}
\oint_{Cy} \frac{\tilde{d}Q}{T} \le 0
\label{(4)}
\end{equation}
From \eqref{(4)}, the concept of temperature as an ``integrating factor'' for the improper ``heat infinitesimals'' $\tilde{d} Q$ arises, as well as the concept of \emph{entropy}, denoted by $S(X)$ , a function of the state $X$, yielding a second form of Clausius' inequality \eqref{(3)}:

\emph{Clausius' Inequality (2)}:

In an adiabatic transformation from the state $X_{1}$ to the state $X_{2}$,
\begin{equation}
S(X_{1}) \ge S(X_{2})
\label{(5)}
\end{equation}

The equality in \eqref{(5)} occurs if and only if the process $C_{12} = X_{1} \to X_{2}$ is reversible.

\eqref{(5)} leads to Clausius' ``sweeping'' (in the words of ter Haar and Wergeland (\cite{tHW}, introduction, p. xiii.) formulation of the second law:

\emph{Second Law (Cl)}: The entropy of the Universe rises to a maximum value.

A drawback of the concept of entropy $S(X)$ is that it derives from \eqref{(4)}, i.e., from the formulation of thermodynamics through the use of differential forms, which requires differentiability of the various state functions, and there are points in state-space where differentiability does not hold, namely, when phase transitions occur (see the discussion in \cite{LYPR}, p.35). The only complete solution to this problem, as far as I know, is given in the axiomatic treatment of Lieb and Yngvason. I proceed, however, with an attempt at a textbook formulation, and assume that an entropy function has been constructed in the previous (usual) manner for simple pure states $X$ (see \cite{ASWight} for this concept), that is, outside multi-phase regions.

In the textbooks (\cite{WreEd}, \cite{Cal}, \cite{tHW}), one defines systems with \emph{internal barriers} or \emph{constrained systems}. Consider a \emph{closed} system consisting of two subsystems, between which an entirely restrictive wall (that is, to the establishment of equilibrium) exists - i.e., it does not allow any exchange of energy, change of volume or number of particles. This system is what is called a \emph{``constrained equilibrium''}. Assume that the first subsystem is in the state $X_{1} = (U_{1},V_{1},N_{1})$, and the second one in the state $X_{2} = (U_{2},V_{2},N_{2})$. One \emph{defines} the entropy of the system in constrained equilibrium by

\begin{equation}
S(U_{1},V_{1},N_{1};U_{2},V_{2},N_{2}) \equiv S(U_{1},V_{1},N_{1}) + S(U_{2},V_{2},N_{2})
\label{(6)}
\end{equation}
After removal of the wall, it is assumed (see later) that an equilibrium state $X = (U,V,N)$ is established. Since further reintroduction of the wall alters nothing whatever in the subsystems' equilibrium,, one may also visualize the equilibrium state $(U,V,N)$ as a special constrained equilibrium, and one may prove:

\begin{theorem}
\label{th:2.1}
Among all constrained equilibria of the same energy, the true equilibrium state has maximum entropy, i.e., the function $S(U,V,N)$ is superadditive:
\begin{equation}
\label{(7.1)}
S(U_{1}+U_{2}, V_{1}+V_{2},N_{1}+N_{2}) \ge S(U_{1},V_{1},N_{1}) + S(U_{2},V_{2},N_{2})
\end{equation}

Furthermore, if the inequality sign in \eqref{(7.1)} is strict, the process is \emph{irreversible}, and the entropy increase $\Delta S$:
\begin{equation}
\label{(7.2)}
\Delta S \equiv S(U_{1}+U_{2}, V_{1}+V_{2},N_{1}+N_{2}) - S(U_{1},V_{1},N_{1}) - S(U_{2},V_{2},N_{2})
\end{equation}

is a measure of irreversibility.

\end{theorem}

\begin{proof}

Consider the process of \emph{removal} of the wall between two subsystems $(U_{1},V_{1},N_{1})$ and $(U_{2},V_{2},N_{2})$, maintaining the internal energy constant. After the equilibrium state $(U,V,N)$ has been attained, couple the system to a series of reservoirs at temperatures $T_{1}, T_{2}, \cdots T_{n}$, in such a way that the compound system performs a cycle, i.e., returns to the initial state. By the argument preceding \eqref{(3)}, this coupling is equivalent to couple the system to a unique heat reservoir at temperature $T_{0}$, say; let $Q_{0}$ denote the amount of heat exchanged with it. By the second law (KP), we find
$$
\frac{Q_{0}}{T_{0}} = S(U_{1},V_{1},N_{1}) + S(U_{2},V_{2},N_{2}) - S(U,V,N) \le 0
$$
which is \eqref{(7.1)}. The second assertion follows from the fact that, if $Q_{0} < 0$, work will have been transformed into heat (energy) without other changes, which is irreversible, again by the second law (KP).

\end{proof}

Together with the process of extensivity, or homogeneity of the first degree,

\begin{equation}
\label{(8)}
S(\lambda U, \lambda V, \lambda N) = \lambda S(U,V,N)
\end{equation}
the property of superadditivity \eqref{(7.1)} leads to the fundamental property of \emph{concavity}:

$S(X)$ is a concave function of $X=(U,V,N)$, i.e., for $0 \le \alpha \le 1$,
\begin{equation}
\label{(9)}
S(\alpha X_{1} + (1-\alpha) X_{2}) \ge \alpha S(X_{1}) + (1-\alpha) S(X_{2})
\end{equation}

We now make the

\emph{Assumption 1}

\eqref{(7.1)}, \eqref{(8)} and \eqref{(9)} are the \emph{basic properties of entropy}, valid also in the presence of phase transitions, i.e., in multi-phase regions.

The far-reaching meaning of assumption 1 is due to the fact (see, e.g., \cite{RV}) that a concave function (on an open set) is only continuous, but not necessarily differentiable: it may exhibit a countable number of points of (finite) discontinuity, which are thus naturally associated to phase transitions. They account for the beautiful theory exposed in \cite{ASWight}, essentially due to Gibbs, of multicritical points, based on the structure of the boundaries of convex sets. It is to be emphasized that Lieb and Yngvason obtain \eqref{(9)} from their set of axioms, which allow to construct the entropy function without assumption 1.  

The \emph{existence} of irreversible processes lies, as remarked in \cite{LYPR}, p. 35, at the very heart of thermodynamics. ``If they did not exist, it would mean that nothing is forbidden, and there would be no second law''. In a proper language, this is one of the axioms of \cite{LYPR} ((S1), p. 42), which is related to Carath\'{e}odory's principle, by Theorem 2.9, p. 35, of \cite{LYPR}. As a brief historical remark, Max Born was one of the few who recognized the importance of Carath\'{e}odory's work, already in the early twenties, see the very readable article by Landsberg \cite{Landsberg} and, as a complementary historical paper on the origin of exact and non-exact differentials in mechanics and thermodynamics, see \cite{Oliveira2}.

My first ``small point'' now follows (as revisited on the light of what I know today). Virtually all real physical processes are irreversible. Is it possible, in this connection, to provide at least one concrete, physically reasonable, example of an irreversible process, and relate it to Clausius' formulation of the second law, in which time is implicit as a parameter?

The next section is reserved to a possible answer to this question. For clarity, I divide it into four parts.

\section{A possible answer and some difficulties}

\subsection{The free adiabatic expansion of the ideal gas}

We now consider $N$ particles of an ideal gas, initially inside a container of volume $V$, which is allowed to expand adiabatically into another container of the same volume $V$ (for simplicity), in which, originally, vacuum had been established (Gay-Lussac experiment, \cite{tHW}, p.36). We use theorem ~\ref{th:2.1} with $U_{2}=0$, $V_{1}=V_{2}=V$, $N_{1}=N_{2}=N$, and find

\begin{equation}
\Delta S = S(U,2V,N) - S(U,V,N)
\label{(10)}
\end{equation}

For the ideal gas, from the two state equations
\begin{equation}
\frac{1}{T} = \frac{3}{2} k \frac{N}{U}
\label{(11.1)}
\end{equation}
and
\begin{equation}
\frac{p}{T} = k \frac{N}{V}
\label{(11.2)}
\end{equation}
with $T$ the absolute temperature and $p$ the pressure, we find
\begin{equation}
S(U,V,N) = Ns(u,v)
\label{(12.1)}
\end{equation}
with $u = \frac{U}{N}$, $v = \frac{V}{N}$. From
\begin{equation}
ds = \frac{1}{T} du + \frac{p}{T} dv
\label{12.2)}
\end{equation}
we obtain
\begin{equation}
s(u,v) = s_{0} + \frac{3}{2} k \log(\frac{u}{u_{0}}) + \log(\frac{v}{v_{0}})
\label{(13)}
\end{equation}
where $s_{0},u_{0},v_{0}$ are constants. By \eqref{(10)} - \eqref{(13)}, and the fact that $u$ is constant in the process, we finally obtain
\begin{equation}
\Delta S = k N \log(2)
\label{(14)}
\end{equation}
This result is reproduced in most textbooks. Is it physically reasonable? When we open a valve and let gas rush into a vacuum, macroscopic motions occur. Although it may seem ``obvious'' that the gas will fill the second container uniformly, when equilibrium is attained, this is not at all evident! Indeed, why does an inhomogeneity at some place in a gas disappear, from a physical standpoint? Intuitively, we expect that it is the collisions between the molecules (think of them as billiard balls) which finally produce a uniform distribution. This intuition is refined in section 4, where it is related to the property of mixing and $K$ systems. But the molecules of a free gas do not interact! More precisely, the free gas is shown there not to satisfy the mixing property. We have, therefore, not provided a ``physically reasonable'' example of an irreversible process, and try again in the next section.

\subsection{The free, adiabatic expansion of a van der Waals gas outside the saturation region}

The simplest model of an interacting gas is the van der Waals gas (see \cite{tHW}, pp. 4-7, or \cite{WreEd}), which is a caricature of a hard-core repulsion. It is described by the equation of state
\begin{equation}
\label{(15)}
p = \frac{kT}{v-b} - \frac{a}{v^{2}}
\end{equation}
where $a > 0$, $b > 0$ are parameters, related to the critical data of the gas. For sufficiently large $T$, the term $\frac{a}{v^{2}}$ and the correction $b$ in \eqref{(15)} may be neglected, and, thus, the specific heat at constant volume $C_{V}$ is close to the ideal gas value $C_{V} = \frac{3}{2} Nk$ resulting from \eqref{(11.1)}; in the sequel we assume
\begin{equation}
C_{V} = constant = C \mbox{ with } 0< C < \infty
\label{(16)}
\end{equation}
which is a good approximation ``sufficiently far'' from the saturation region: \eqref{(16)} is rigorously controllable by suitable bounds. In this region, differential forms (see, e.g., \cite{Dor}, p.12) may be used freely. Using \eqref{(16)}, we find immediately, for the reversible process by which the gas initially occupies the volume $V_{i}=V$ and expands to the final volume $V_{f}=2V$, the condition of constant internal energy
\begin{equation}
\label{(17)}
dU = C dT + \frac{aN}{v^{2}} dv = 0
\end{equation}
from which
$$
dT = - \frac{aN}{C} \frac{dv}{v^{2}}
$$
The initial temperature $T_{0}$ (with $0 < T_{0} < \infty$) is fixed by the initial (constant) energy: $U(V,T_{0}) = C T_{0} - \frac{a}{V_{i}}$ (up to a constant which we may fix as zero), and we obtain
\begin{equation}
T = T_{0} + \frac{aN}{Cv}
\label{(18)}
\end{equation}
From \eqref{(17)} and \eqref{(18)},
$$
dS = \frac{p}{T} dV = k N \frac{dv}{v-b} - \frac{aN}{(\frac{aN}{Cv}+T_{0})v^{2}} dv
$$
from which
\begin{equation}
S(N,v) = kN \log(v-b) + \frac{aN}{\mu} \log(1+ \frac{\mu}{T_{0}v}) + \mbox{ const. }
\label{(19)}
\end{equation}
where
\begin{equation}
\mu \equiv \frac{aN}{C}
\label{(20)}
\end{equation}
Finally, we find
\begin{equation}
S(N,2v) - S(N,v) = S_{1}(N,v) + S_{2}(N,v)
\label{(20.1)}
\end{equation}
where
\begin{equation}
S_{1}(N,v) \equiv k N \log(\frac{2v-b}{v-b})
\label{(20.2)}
\end{equation}
and
\begin{equation}
S_{2}(N,v) \equiv C \log(\frac{1+\frac{\mu}{2T_{0}v}}{1+\frac{\mu}{4T_{0}v}})  
\label{(20.3)}
\end{equation}
$S_{1}$ is the Boltzmann ``probabilistic term'', which does not depend on $a$ and tends to \eqref{(14)} as $b \to 0$. $S_{2}$ independs of $b$ and tends to zero as $a \to 0$. Thus, $S(N,2v)-S(N,v) \to kN\log(2)$, which is the ideal gas result \eqref{(14)}, as expected, as $b \to 0$ and $a \to 0$. Moreover, $S_{2}$ may be interpreted as an ``interaction term'', which is strictly positive, because
\begin{equation}
S_{2}(N,v) > 0 \mbox{ if } \mu > 0
\label{(20.4)}
\end{equation}

It is interesting to observe that the inclusion of interaction leads to a measurable effect, i.e., the cooling observed in the Gay-Lussac experiment (see \eqref{(18)} and \cite{tHW}, p. 37).

This exercise becomes much more difficult in the saturation region! There, Assumption 1 leads to the construction of $S(U,V,N)$, by Lebowitz and Penrose \cite{LP1}, but \eqref{(16)} no longer holds, of course, and the differential forms, in general, are not well-defined.

The structure of \eqref{(20.1)}, \eqref{(20.2)}, \eqref{(20.3)} shows that the probabilistic term and the interaction term do not interfere in the entropy variation function. This is due to the fact that dynamics, which reduces to the process of lifting the barrier (or opening the valve), is \emph{independent} of the internal energy of the system and hence of the state. This looks like a drastic approximation, and, indeed, it is! In the next section, I try to expand on this point. 

\subsection{Dynamics, or no dynamics, and what to expect}

What is, then, the answer to the question posed at the end of section 2? A \emph{tentative} answer is that both examples in subsections 3.1 and 3.2 are not physically reasonable dynamically, i.e., if we wish to go on with second law (Cl), in which time is implicit as a parameter. There are at least two solutions of this dilemma:
\begin{itemize}
\item [($a.)$] No dynamics. In this case we have either of two possibilities:
\item [($a1.)$] Existence of irreversible processes as a (crucial) assumption in \cite{LYPR}, (S1), p.42, as previously mentioned. With it, however, a full theory is developed, which, in particular, does not use differential forms, and is thus extendible to multi-phase regions;
\item [($a2.)$] ``concrete'' existence of irreversible upon inclusion of internal barriers (constrained systems) in the equilibrium formalism, with, however, the use of differential forms, i.e., excluding multi-phase regions.
\end{itemize}

It goes without saying that the levels of mathematical rigor in $a1$ and $a2$ are quite different, $a2$ being at a much lower (elementary textbook) level.
\begin{itemize}
\item [$(b.)$] Dynamics.
\end{itemize}
In this case we have ``concrete'' existence of irreversible processes, if we make use of the possibility of considering (as in \cite{WreEd}) a system with internal barriers as a \emph{special} example of a \emph{nonequilibrium state}. Indeed, these barriers are an example of a ``sudden'' interaction, which introduces an instantaneous global change in the system - i.e., the gas expands \emph{instantaneously}, filling the total volume uniformly. These interactions contain a ``delta function'' in the time variable, which change the energy by an \emph{infinite} amount. They are, therefore, physically inadmissible, but may be regarded as a limiting case of certain physically admissible interactions, according to which we expect that the gas will eventually (after a relaxation time) fill the total volume uniformly in the free adiabatic vacuum expansion (Gay-Lussac experiment). This means going beyond the soluble example of section 3.2, for which purpose we need an

\emph{Assumption 2}  Any initial state $X_{0} \equiv (U_{0},V_{0},N_{0})$ approaches an equilibrium state $X_{\infty} \equiv (U_{\infty},V_{\infty},N_{\infty})$ as $t \to \infty$.

By the second law (Cl), $X_{\infty}$ should be a maximum of the entropy. This statement should be understood in the sense of theorem ~\ref{th:2.1}: for an arbitrary division of the system by an impenetrable wall such that \eqref{(6)} holds, the equilibrium entropy is a maximum in the sense that:
\begin{itemize}
\item [($c.)$] the variation $\delta S(X_{\infty}) = 0$;
\item [($d.)$] for any constrained equilibrium $X \ne X_{\infty}$, 
\begin{equation}
\Delta S(X) > 0
\label{(21.1)}
\end{equation}
\end{itemize}
Above, the variation $\delta S(X_{\infty})$ is defined within the class of constrained equilibria, that is,
$$
\delta S(X_{\infty}) \equiv \frac{d}{dt} [S_{1}(U_{1,\infty}+tU,V_{1,\infty}+tV) + S_{2}(U_{2,\infty}-tU,V_{2,\infty}-tV)]_{t=0}
$$
for $U,V$ arbitrary, $U_{\infty}=U_{1,\infty}+U_{2,\infty}$ and $V_{\infty}=V_{1,\infty}+V_{2,\infty}$. Using $\frac{\partial S}{\partial U}=\frac{1}{T}$, together with $\frac{\partial S}{\partial V} = \frac{p}{T}$, we obtain as equilibrium conditions
\begin{equation}
\label{(21.2)}
T_{1} = T_{2}
\end{equation}
as well as
\begin{equation}
p_{1} = p_{2}
\label{(21.3)}
\end{equation}
Since the above hold for any subdivision, \eqref{(21.3)} implies, of course, that the gas fills the total container uniformly. 

The quantity $\Delta S$ in \eqref{(21.1)} is defined by \eqref{(8)} of theorem ~\ref{th:2.1}, with $(U,V,N)$ there identified as $(U_{\infty},V_{\infty},N_{\infty})$. Thus, item $d$ above means that any transition from the given equilibrium state to a different constrained equilibrium which occurs spontaneously, i.e., without changes in the surroundings, is forbidden by the second law {KP). When d.) is satisfied, the equilibrium is called \emph{stable}.

\begin{remark}
\label{Remark 3.1}
Since every subsystem of an equilibrium state is itself in equilibrium, i.e., the insertion of any impermeable wall does not change the equilibria of the subsystems, it is possible to describe the thermodynamics of a multi-phase system by its pure phases \cite{ASWight}.
\end{remark}

\begin{remark}
\label{Remark 3.2}
Due to the fact that systems with internal barriers comprise only a (very) special class of perturbations of equilibrium states, definition $d$ of a stable equilibrium is not quite complete: according to the physical situation, several types of stability might be envisaged, ranging from local to nonlocal perturbations.
\end{remark}

The understanding of dynamics is, of course, essential to grasp the (real) nonequilibrium phenomena observed in Nature. There are several approaches to this important issue: see \cite{LYNE} and references given there, as well as the paper by Abou-Salem and Fr\"{o}hlich \cite{AF} and references given there. If we proceed along the lines of Assumption 2, progress may require the interplay of thermodynamics with other areas of physics, a subject to which we now turn.

\section{The interplay between thermodynamics and other areas of physics: dynamical systems and statistical physics}

\subsection{The picture furnished by the theory of dynamical systems}

Consider the dynamical system generated by the motion of $N$ matter points (``gas molecules'') in a fixed volume $V$, and let $\Gamma=\{x\equiv(q_{1},p_{1}), i=1, \cdots 3N \}$ denote the corresponding phase space, with $q_{1}$ being the generalized coordinates and $p_{i}$ the momenta. If the system is isolated (fixed total energy $E$), $\Gamma$ will be compact, and there exists an invariant measure $\mu$ which may be interpreted as the distribution in thermodynamic equilibrium \cite{Sz}. Starting from initial conditions in a certain volume $V_{1}$, the set of admissible coordinates and momenta of the molecules forms a subset $A \in \Gamma$. The formula
\begin{equation}
\mu_{0}(B) = \frac{\mu(A \cap B)}{\mu(A)}
\label{(22.1)}
\end{equation}
defines a \emph{initial} distribution (which is not the equilibrium distribution), interpreted as \emph{conditional} distribution relatively to the system's known initial condition. We may now relate $\mu_{0}$ to the knowledge of the state of the system at time $t$ by the measure $\mu_{t}$ defined by
\begin{equation}
\mu_{t}(B) = \mu_{0}(T_{t}(B))
\label{(22.2)}
\end{equation}
where $T_{t}$ is the (Hamiltonian) time evolution (flow) which leaves $\mu$ invariant. In the given example
\begin{equation}
\mu(F) = \mu_{E}(F) = \int_{F} \delta(H(x)-E) dx
\label{(22.3)}
\end{equation}
is the \emph{microcanonical Gibbs measure}, where $H(x)$ is the Hamiltonian describing the system and \eqref{(22.2)} is \emph{Liouville's theorem}. Suppose, now, that $\mu_{0}$ is absolutely continuous (a.c.) with respect to $\mu$, i.e., there exists $\rho_{0} \in L^{1}(M=\Gamma,d\mu)$ (integrable w.r.t. phase-space (Lebesgue) measure) such that
\begin{equation}
\label{(22.4)}
d\mu_{0} = \rho_{0}(x) d\mu
\end{equation}
Then,
\begin{equation}
\mu_{t}(B) = \int_{M} \chi_{B}(T_{t}x)d\mu_{0} = \int_{M} \chi_{B}(T_{t}x)\rho_{0}(x)d\mu
\label{(22.5)}
\end{equation}
We may define the property of \emph{mixing} by the relation 
\begin{equation}
\lim_{t\to\infty}\mu_{t}(B) = \int_{M} \chi_{B}d\mu \int_{M}\rho_{0}d\mu = \mu(B)1 = \mu(B)
\label{(22.6)}
\end{equation}
which means that, whatever the initial distribution, normalized and a.c. w.r.t. $\mu$, the time-translates $\mu_{t}$ of $\mu_{0}$ under $T_{t}$ converge, for $t \to \infty$, to the equilibrium distribution. Note that the \emph{existence} of the limit on the l.h.s. of \eqref{(22.6)} is a part of the assumption: \eqref{(22.6)} may be taken as the definition of the \emph{approach to equilibrium}. This subject is treated in much greater detail and depth, of course, in the articles and book by Penrose \cite{Pen1}, \cite{Pen}, to which I must refer for a better understanding of the concepts introduced here, as well as a wealth of applications, and further references. A simple, illuminating introduction is provided by the article \cite{LP}.

What does \eqref{(22.6)} mean? It means that mixing systems are ``memoryless'', i.e., they posess a stochastic character which justifies the probabilistic framework of equilibrium statistical mechanics presented at the introductory texts (for a recent, specially clear and pedagogical exposition, see \cite{Salinas}). The microscopic mechanism of this loss of memory is the sensitive (exponential) dependence on initial conditions produced by ``defocalizing shocks'' between the gas molecules, first pointed out by Krylov (see \cite{Sinai1} and references given there).

We may picture the gas molecules as a system of hard spheres enclosed in a cube with perfectly reflecting walls or periodic boundary conditions. this is supposed to be a $K$-system (see \cite{Wa}, Definition 4.7, p.101). Rising still one step in this so-called ergodic hierarchy (see \cite{LP}), we come to \emph{Bernouilli systems}, such as the one we presently introduce.

Define $T_{2}$ by $T_{2} : X \to Y$, where $X=[0,1]$ and $Y=[0,1]$, by
\begin{equation}
T_{2}x = fr(2x) \equiv 2x \mod 1
\label{(23.1)}
\end{equation}
Above, $fr(x) = x-[x]$, where $[x]$ denotes the largest integer which is smaller or equal to x. Note that this mapping is not one-to-one (it is a so-called endomorphism). In fact, we see that the inverse image of a point $x$ is either $\frac{x}{2}$ or $\frac{x+1}{2}$. $T_{2}$ is called the dyadic transformation and leaves Lebesgue measure $\mu$ (on the line) invariant, because the inverse image of a point $y$ is
\begin{equation}
T_{2}^{-1}(y) = \{\frac{y}{2}\} \cup \{\frac{y+1}{2}\}
\label{(23.2)}
\end{equation} 
Indeed, from \eqref{(23.2)}, $\mu(T_{2}^{-1}([0,y]))= y$, which generalizes to
\begin{equation}
\mu(A) = \mu(T_{2}^{-1}(A)) \mbox{ for any Borel set } A \in [0,1]
\label{(23.3)}
\end{equation}
$T_{2}$ has a simple alternative description. Let $x$ be given by its expansion in basis $2$, i.e., 
\begin{equation}
x = .\epsilon_{1} \epsilon_{2} \cdots = \sum_{n=1}^{\infty} \frac{\epsilon_{n}}{2^{n}} \mbox{ with } \epsilon_{n} \in \{0,1\}
\label{(23.4)}
\end{equation}
Then, it is easy to see that
\begin{equation}
T_{2}x = .\epsilon_{2}\epsilon_{3} \cdots \mbox{ and in general } T_{2}^{n}x = .\epsilon_{n+1}\epsilon_{n+2}, \cdots
\label{(23.5)}
\end{equation}
that is, $T_{2}$ is the ``one-sided shift'' in this representation. Given a Lebesgue integrable function, i.e.,
\begin{equation}
f \in L^{1}(0,1)
\label{(23.6)}
\end{equation}
we may define the Ruelle-Perron-Frobenius operator $P$ (\cite{LaM}, \cite{MWB}) by
\begin{equation}
\int_{A} (Pf)(x)dx = \int_{T_{2}^{-1}(A)} f(x)dx
\label{(23.7)}
\end{equation}
If
\begin{equation}
f_{0}(x) \equiv 1
\label{(23.8)}
\end{equation}
it follows from \eqref{(23.4)} and \eqref{(23.7)} that 
\begin{equation}
Pf_{0} \equiv 1
\label{(23.9)}
\end{equation}
Let $A=[0,x]$. From \eqref{(23.7)}, it follows that
\begin{eqnarray*}
(Pf)(x) = \frac{d}{dx} (\int_{0}^{\frac{x}{2}} f(u) du + \int_{\frac{1}{2}}^{\frac{x+1}{2}} f(u)du)=\\
= \frac{1}{2} [f(\frac{x}{2})+f(\frac{x+1}{2})]
\end{eqnarray*}
from which, by iteration,
\begin{equation}
(P^{n}(f))(x) = \frac{1}{2^{n}} \sum_{k=0}^{2^{n}-1} f(\frac{x+k}{2^{n}})
\label{(24.1)}
\end{equation}
We shall say that $f$ is a \emph{density} if $f \ge 0 \mbox{ a.e. } \in (0,1)$ and $\int_{0}^{1}f(x)dx = 1$; $f_{0}$, given by \eqref{(23.8)}, is the \emph{uniform density}. The words $a.e.$ stand for almost everywhere, that is, in the complement of a set with (Lebesgue) measure zero. By \eqref{(23.7)} and \eqref{(23.9)}, $P$ maps densities to densities, and, by \eqref{(24.1)},
\begin{equation}
\lim_{n\to\infty}(P^{n}f)(x) = \int_{0}^{1} f(y) dy = 1
\label{(24.2)}
\end{equation} 
Together with \eqref{(23.9)}, \eqref{(24.2)} is a form of the property of approach to equilibrium \eqref{(22.6)} (see also \cite{LaM}, Theorem 4.4.1, p.65): the evolution of any density tends (for (discrete) time $n \to \infty$) to the unique invariant density, the uniform density $f_{0}$ given by \eqref{(23.8)}.

Of particular importance above is that $P$ is not defined on individual orbits ot the map, which would correspond to taking $f$ in \eqref{(23.7)} to be a ``delta function'', which is not allowed by condition \eqref{(23.6)}. Indeed, these individual orbits behave rather erratically. Consider two points $x_{1},x_{2}$, both in $[0,1]$, close in the sense that the first $n$ digits in the expansion \eqref{(23.4)} are identical. By \eqref{(23.5)}, it follows that $T_{2}^{n}x_{1}$ and $T_{2}^{n}x_{2}$ differ already in the first digit: an initially exponentially small difference $2^{-n}$ is magnified by the evolution to one of order $O(1)$. When this property holds for $n$ arbitrarily large, as in the present example, one speaks of the \emph{exponential sensitivity to initial conditions} mentioned before in connection with Krylov's mechanism. This is the aforementioned stochastic element: for almost all $x_{0}$ (in the sense of Lebesgue measure, i.e., excluding a set $\gamma$ of zero Lebesgue measure), $T_{2}^{n}x_{0}$ comes arbitrarily close to almost any $x \in [0,1]$ if $n$ is taken arbitrarily large, or, in other words, it ``fills'' the whole interval uniformly throughout the evolution, and, thus, $\lim_{n\to\infty} T_{2}^{n}x_{0}$ does not exist for a.e. $x_{0}$. 

The set $\gamma$ consists of the finite dyadic numbers $x_{f}$, i.e., those whose dyadic expansion \eqref{(23.4)} is finite; it is immediate from the definiton of $T_{2}$ that $T_{2}^{n}x_{f} \to 0 \mod 1$, the latter being the fixed points of the dyadic map: they are untypical, in the sense that they do not fill the interval uniformly. They may be analogous to some ``bad'' initial configurations of the gas (not! initial states $(U,V,N)$: a configuration is a set of values of position and momentum coordinates of the $N$ particles in the gas inside the volume $V$, such that the total internal energy is $U$), for example, those particle configurations with all initial velocities directed away from the barrier which is lifted. Their measure in phase space (in three dimensions) is also zero, similarly to the set $\gamma$.

As a final remark, to connect with section 3.1, a map with the property of mixing is necessarily ergodic (\cite{Pen}, \cite{Pen1}, \cite{LP}), whose definition may be taken to be: a flow $T_{t}$ is ergodic if and only if any invariant function $\Phi$ under the flow (i.e., such that $\Phi(T_{t}(x))=\Phi(x)$ for a.e.$x$ is a.e. a constant. Thus, assuming the microcanonical measure \eqref{(22.3)}, any function $\Phi$ invariant under the flow is a functional of the total Hamiltonian $H(x)$. For a free system, however, $H = \sum_{i=1}^{N} H_{i}$, with $N \ge 2$ (at least two particles), and \emph{each} $H_{i}, i=1, \cdots, N$ is invariant under the flow. Thus, a free system is not ergodic, and, therefore, not mixing.

\subsection{Connections with statistical physics}

We now discuss additional points of the possible connection between the previous discussion and the vacuum expansion of a gas of a large number $N$ of molecules. We have seen that, for $T_{2}$, the evolution of densities does approach a limit, in contrast, in general, to that of individual orbits. This means that a ``coarse-graining'' in the space variables implements a kind of ``restoring force'' which pushes toward equilibrium: a density is defined by its values on an \emph{infinite} number of points. This phenomenon is, for a deterministic system, much subtler than the analogous one in an \emph{a priori} probabilistic context. One example of the latter is the Ehrenfest urn model in the elementary theory of Markov chains (\cite{Chu1}, Example 16, p.283).

Remaining in the deterministic framework, the above mentioned coarse-graining is analogous to consider an \emph{infinite} number of molecules, i.e., to take the thermodynamic limit $N \to \infty$, together with $V \to \infty$, with $\frac{N}{V} = \rho$, the particle density, taken to be a constant. A possible approach to prove the analogue of Assumption 2 follows the ideas of \cite{therm2}.

Instead of the entropy $S(U,V,N)$ the primary object is the \emph{specific entropy}, i.e., the entropy per unit particle (or unit volume) in the thermodynamic limit (assuming it exists), which in the previous discussion would correspond to
\begin{equation}
s(u,v) = \lim_{N,V \to \infty;\frac{V}{N}=v;\frac{U}{N}=u} \frac{1}{N} S(U,V,N)
\label{(25.1)}
\end{equation}
The thermodynamic limit performs a ``coarse graining'' in the particle system, which is analogous to considering the evolution of densities, instead of individual trajectories, in dynamical systems. One might therefore expect that, starting from initial values $(u_{0},v_{0})$ and an initial specific entropy 
\begin{equation}
s(u_{0},v_{0}) \equiv s_{0}
\label{(25.2)}
\end{equation}
and setting $u_{t}=T_{t}u_{0}$, $v_{t}=T_{t}v_{0}$, where $T_{t}$ is the flow in \eqref{(22.2)},

\emph{Modified Second Law(Cl)}
\begin{equation}
s_{max} = \lim_{t \to \infty} s(u_{t},v_{t})
\label{(25.3)}
\end{equation}

The (analogue of the) modified Second Law (Cl) has been demonstrated for a class of quantum systems in \cite{therm2}, with the nontrivial feature that
\begin{equation}
\label{(25.4)}
s_{max} > s_{0}
\end{equation}
In \cite{therm2}, the von Neumann definition of entropy 
\begin{equation}
S_{vn}(N,V) = -k tr \rho_{N,V} \log \rho_{N,V}
\label{(26.1)}
\end{equation}
where $\rho_{N,V} \ge 0$ is an operator of unit trace on the Hilbert space of a quantum system of $N$ particles in a volume $V$ (see \cite{Wehrl1}), i.e.,
\begin{equation}
tr \rho_{N,V} = 1
\label{(26.2)}
\end{equation}
The classical analog of $S_{vn}$ is the Gibbs entropy $S_{G}$:
\begin{equation}
S_{G}(N,V) = -k \int_{\Gamma} dx \rho(x) \log \rho(x)
\label{(26.3)}
\end{equation}
in the notation of \eqref{(22.1)} et seq., with $\rho$ being the density of the invariant measure $\mu$ , assumed to be absolutely continuous w.r.t. Lebesgue measure $m=dx$ on phase space $\Gamma$ corresponding to the distribution in thermodynamic equilibrium, i.e., 
\begin{equation}
\mu(B) = \int_{B} \rho(x) dx
\label{(26.4)}
\end{equation}
where $B$ is any (Borel) subset of $\Gamma$.

\subsection{The Modified Second Law (Cl) and the problem of proving Assumption 2}

The theorem proved in \cite{therm2} renders mathematically precise a result of Gibbs \cite{Gibbs}, as reformulated by Penrose \cite{Pen1}. It relies on the fact that the specific entropy is not continuous as a function of $t$, but rather only \emph{upper semicontinuous}. The significance of this property (together with its precise definition) is well discussed in the classic book by Geoffrey Sewell (\cite{Sewell1}, 3.2.3). The associated Modified Second Law (Cl) given in the previous subsection seems also quite natural from the physical point of view, because it is the specific entropy that is measured. an example of this is to be found in section 5, where the analogous quantity (entropy per unit volume) is related to the number of photons per unit volume in the model of the cosmic microwave background radiation (CMB) which pervades the Universe.

The thermodynamic limit reflects a coarse graining in the space variables, as discussed in the previous section, whose importance is as crucial as the corresponding operation in the theory of dynamical systems (see \eqref{(24.2)}). The reason is that, in general, the thermodynamic limit does not commute with the long-time limit in \eqref{(25.3)}, and it is precisely this feature which enables \eqref{(25.4)}. For finite $N,V$, the Penrose-Gibbs theorem (\cite{Pen1}, p.1959) yields equality in \eqref{(25.4)}, because the full entropy $S(U_{t},V_{t})$ (for fixed $N$) is continuous rather than upper semicontinuous. This noncommutativity will reappear (explicitly!) in a class of models for the cosmological evolution of the photons in the CMB in section 5 (the forthcoming proposition 5.1). 

The problem remains, however, to prove Assumption 2 in a wider class of models, classical and quantum. This is the famous problem of the approach to equilibrium in closed systems. The greatest progress for quantum continuous systems is due to Narnhofer in Thirring, for a model of interacting fermions \cite{NTh4}.

Summarizing: the central issues in the proof of the Modified Second Law (Cl) are:
\begin{itemize}
\item [$a.)$] the initial state $X \ne X_{eq}$;
\item [$b.)$] a coarse graining in the space variables, represented by considering the specific entropy rather than the entropy;
\item [$c.)$] the approach to equilibrium for the state must be proved (Assumption 2);
\item [$d.)$] the time evolution is deterministic. 
\end{itemize}

\subsection{Irreversibility and the time arrow}

As remarked by Griffiths \cite{Gr}, since two bodies at unequal temperatures, but in thermal contact, are known to exchange energy in such a way that the temperatures approach each other even in the presence of a magnetic field, which breaks time-reversal invariance, the latter is certainly not the key for understanding macroscopic irreversibility.

Most of the states in Nature are unstable (nonequilibrium) states. In atomic and molecular physics, for instance, all states with the exception of the ground state are resonances, in particle physics, all but the lightest particles are unstable.

Assume an initial (unstable) state at $t=t_{0}$, of a given atomic resonance, together with the electromagnetic field, which decays by emission of a photon, to a final state of the composed system (see \cite{Wreszunst} and references given there). The initial unstable state must have been prepared at some (finite) time in the past, i.e., at some time $-\infty < t_{p} < t_{0}$, by some process (e.g., resonant scattering), which delivered to the (compound) system a \emph{finite} amount of energy. We have suggested in \cite{therm2} and \cite{Wreszunst} that this preparation of the state is not time-reversal invariant: this is the reason for the existence of an \emph{arrow of time}. 

A beautiful discussion of irreversibility is in section 3.8 of Peierls' book \cite{Peierls}. On the first paragraph of p. 76, he makes a statement in the same spirit of the above-mentioned preparation of the state. He further adds: ``We thus see that the asymmetry arises, not from the laws governing the motion, but from the boundary conditions we impose to specify our question''.

Once given a definite time direction, the Modified Second law (Cl) implies \emph{irreversibility}, as long as \eqref{(25.4)} holds.

\section{The Universe and the interplay between thermodynamics and field theory: the second small point}

Together with the Second law (Cl), Clausius formulated the first law in the form

\emph{First law (Cl)} The energy of the Universe is constant.

As a student, I knew that, when studying the large-scale dynamics of the Universe, the basic tool is classical field theory (in the form of classical general relativity, see \cite{Lud} as a lucid intermediate text, and \cite{Thirr3rd} for an advanced treatment). 

Dyson \cite{Dyson} has studied in detail whether there are any conceptual arguments which require that the gravitational field be quantized, and, on the basis of the classic paper by Bohr and Rosenfeld \cite{BohrR}, arrived at the answer: no. He then investigated whether, experimentally, the graviton may be undetectable, and arrived at the answer: yes, with great probability. There is, therefore, great probability that, experimentally, classical general relativity will provide a description of the physical world which is indistinguishable from the outcomes of a prospective ``quantum gravity''. This remark also applies to the interaction between quantum particles, e.g., photons, and the Universe: there is a large probability the results of the (forthcoming) ``semiclassical'' description are not distinguishable from those from a (hypothetical) fully quantum theory, whose existence remains doubtful.

In general relativity, however, the space-time metric tensor $(g_{ab})$ depends on space and time, and, thus, the time-homogeneity does not hold. Thus, Noether's theorem (\cite{Thirr3rd}, 8.1.5., p. 331) is not applicable, and First law (Cl) does not hold!

Although, in First Law (Cl) the word ``Universe'' was (probably) symbolic for any closed system, the Universe happens to be a paradigm of a (or the) closed system in thermodynamics, and this failure - my ``second small point''- seemed, at the time, to be a catastrophe, in particular because of the supposedly universal character of thermodynamics.

In this section, I attempt to review what is known both on the first and second laws of thermodynamics for the Universe on the light of present knowledge - still keeping with textbook level.

Before considering large-scale properties of the Universe, it is useful to consider non-relativistic models of cosmic sized systems. One example is the system of $N$ neutral massive fermions, interacting via Newtonian forces. This is a model of a neutron star (pulsar), for which the energy $U = -c N^{7/3}$, where $c$ is a positive constant (see \cite{MaRo}, Sec. 2.2.3, for a textbook discussion, \cite{HNT} for a rigorous derivation and \cite{Sewell2} for a comprehensive review). Due to this non-extensivity of the energy, as a consequence of the long-range and attractive charater of the gravitational interaction, the property of concavity of the entropy $S$ \eqref{(9)} does not follow as before and, indeed, there is a regime in which $S$ is \emph{convex} with respect to $U$, and the system undergoes a phase transition of the van der Waals type (see \cite{HNT} and \cite{Sewell2}). In this regime, the specific heat $C_{V} < 0$, mirroring a stage in the stellar evolution in which the star has exhausted the fuel that would burn at that temperature: its core then contracts and heats up, while energy is liberated to the surface, which expands and becomes cooler. Since the star may be considered as an isolated object, this process corresponds to one in which heat (energy) passes by itself from a colder to a hotter body, violating one of the forms of the second law (also due to Clausius). It is, thus, in general, not easy to define entropy, and its increase, in the absence of extensivity, but, remarkably, an extension of the treatment of \cite{LYPR} does allow this: see \cite{LYEM}.

Parenthetically to the above discussion, and in conformance with the property of isolation, the statistical mechanical description, as given, of the aforegoing process of star collapse is restricted to the \emph{microcanonical ensemble}: the result in the canonical ensemble is different. In particular, in the latter ensemble, the specific heat is expressed as an energy fluctuation, which is always positive.

Beyond a certain value of $N$, the non-relativistic models of stars become unphysical, because the mean particle velocities attain values comparable with the velocity of light, and general relativity comes into play. We now turn to this case.

The \emph{cosmological principle} - large-scale spatial homogeneity and isotropy of the Universe - implies that the Universe is ruled by the general (so-called Robertson-Walker) metric (\cite{Thirr3rd}, 10.4.2), \cite{Lud}), which is defined by the ``element of arc''
\begin{equation}
ds^{2} = c^{2}dt^{2}-[R(t)]^{2} \frac{\sum_{a=1}^{3} dx^{a}dx^{a}}{[1+\frac{\kappa}{4}\sum_{a=1}^{3}x^{a}x^{a}]^{2}}
\label{(27.1)}
\end{equation}
Above, $t$ denotes the ``cosmic time'' (henceforth referred to simply as ``time''), $R(t)$ is the so-called \emph{scaling factor} (see \cite{Lud}, whose general lines we adopt), and $\kappa$ denotes the curvature. We fix a unique time $t$ by taking $t=0$ at the ``big bang'', with $\lim_{t \to 0+} R(t)=0$. It is believed that the (present) Universe is flat (zero curvature), with positive so-called cosmological constant $\Lambda$, in which case the exact solution of Einstein's equations is
\begin{equation}
R(t) = [\frac{3}{2} \frac{c}{\Lambda} (\cosh(ct(3\Lambda)^{1/2})-1]^{1/3}
\label{(27.2)}
\end{equation}
(see, e.g., \cite{Pathria}, p.279): it is an ever-expanding Universe, for which
\begin{equation}
\lim_{t \to \infty} R(t) = \infty
\label{(27.3)}
\end{equation}
For a given value of $\Lambda = \Lambda_{c}$, one of the solutions of Einstein's equations is purely static, corresponding to a constant value of $R$, defining the \emph{Einstein model}
\begin{equation}
R = R_{c} = (\Lambda_{c})^{-1/2}
\label{(27.4)}
\end{equation}
(see \cite{Pathria}, p. 278). Denoting the corresponding space metric by $\hat{g_{ab}}$, we have, by \eqref{(27.1)},
\begin{equation}
\hat{g_{ab}} = \frac{g_{ab}}{R(t)^{2}}
\label{(27.5)}
\end{equation}
where we made, for simplicity, a multiplicative renormalization $R(t) \to \Lambda_{c}^{-1/2} R(t)$.

The background radiation (CMB radiation, for ``cosmic microwave background'') fills the Universe today with a black-body spectrum with temperature
\begin{equation}
\hat{T}_{p} \approx 2,7 K
\label{(27.6)}
\end{equation}
(see \cite{Lud} and references given there), where the subscript $p$ stands for ``present''. By the ``hot big bang theory'' \cite{Lud}, it presumably arose from an \emph{equilibrium} state of the radiation field and a plasma of protons and electrons, at a time $t_{0}$ when the (equilibrium) temperature was $\hat{T}(t_{0}) = \hat{T}$, about $3000 K$, the ionization energy of the hydrogen atom (with $\hbar = c = 1$), at a certain $R(t_{0})$. 

A preliminary model of CMB radiation assumes that the metric $g_{ab}$ and the energy-momentum tensor of the photons are \emph{independent}, i.e., do not interact. For the photons, this has the rather drastic consequence that the averages of the energy and momentum at temperature $\hat{T}$, at the cosmic time $t_{0}$, may be related to the same quantities at a later time, by the forthcoming relations \eqref{(27.7)}-\eqref{(27.11)}. At the end of the present section we shall see that modifications of the first law for the Universe do require interaction of the metric with the energy-momentum of the photons (as well as with other matter in the Universe)- the so-called ``back-reaction''. It is, therefore, very surprising that this rather primitive model has had such enormous success: Planck's law fits the Universe's signal with astonishing precision (see, e.g., \cite{Straumann}). As we shall see, the CMB spectrum is, however, highly stable (metastable), and this fact may account for the unexpected accuracy of the approximation, which should probably be regarded within the framework of a ``cosmological perturbation theory'' \cite{Straumann}, shown to be quite successful in the treatment of the anisotropies of the CMB radiation, but no rigorous results exist in this direction. Other examples, not accessible to experiment, such as the thermalization of a quantum field in the presence of an event horizon (the Hawking thermal radiation phenomenon, see \cite{Sewell2} and references given there), show clearly that drastic changes of a quantum field certainly may occur even in a semiclassical picture. 

Even if back-reaction is accounted for, it must be said that unsolved problems remain concerning ``freezing'' the metric and searching for a particle interpretation at each instant of the cosmological time. In the (present) case of a non-stationary metric, the particle interpretation constantly changes with time. The Klein paradox (see, e.g., \cite{Sakurai}, section 3.7, pp 120-121) illustrates this statement well, in the sense that even an agreement on how many particles are being counted depends on each instant of time, since the metric acts like an external potential! 

For the value $R(t_{0})$, we consider a cube of (spatial) volume $\hat{V}=\hat{L}^{3}$, upon which we impose periodic boundary conditions (b.c.) on the radiation field. the wave vectors are $\hat{k}=\frac{2\pi}{\hat{L}}n;n \in \mathbf{Z}^{3}$. the number of photons per unit volume inside a tiny cube in $\hat{k}$-space, which we label by the vertices $(\hat{k}_{1},\hat{k}_{2},\hat{k}_{3}), (\hat{k}_{1}+\Delta \hat{k}_{1},\hat{k}_{2}+\Delta \hat{k}_{2},\hat{k}_{3}+\Delta \hat{k}_{3}$ of $\hat{k}-$ -volume $\Delta_{3}(\hat{k}) \equiv \Delta \hat{k}_{1}\Delta \hat{k}_{2}\Delta \hat{k}_{3}$ is equal to
\begin{equation}
\hat{n}_{t_{0}}(\hat{k}) = \frac{\Delta_{3}(\hat{k})}{\exp(\frac{|k|}{\hat{T}})-1}
\label{(27.7)}
\end{equation}
Above, $|k|= \sqrt{k_{1}^{2}+k_{2}^{2}+k_{3}^{2}}$ is the energy (frequency) of the photons. We shall write henceforth $k$ for $|k|$, for brevity.
Consider, now, any $t>t_{0}$. How does the evolution affect $\hat{n}_{t_{0}}(\hat{k})$?. This is obtained from \eqref{(27.5)}, with the hat denoting the evolution by a static Universe with $R(t)=R(t_{0}) = \mbox{ constant }$, and imposing the equality
\begin{equation}
\hat{n}_{t_{0}}(\hat{k}) \hat{V} = n_{t}(k) V
\label{(27.8)}
\end{equation}
with $\hat{V}=\hat{L}^{3}$, $V=L^{3}$, to obtain
\begin{equation}
r^{-3} \hat{n}_{t_{0}}(\hat{k}) = n_{t}(k)
\label{(27.9)}
\end{equation}
with
\begin{equation}
\frac{\hat{k}}{k} = r \equiv R(t_{0})^{-1}R(t)
\label{(27.10)}
\end{equation}
\begin{equation}
\hat{L} = r^{-1} L
\label{(27.11)}
\end{equation}
Equation \eqref{(27.10)} is the expression of the \emph{cosmological red shift}. By \eqref{(27.7)}, \eqref{(27.9)}, \eqref{(27.10)} and \eqref{(27.11)}, we obtain
\begin{equation}
n_{t}(k) = \frac{\Delta_{3}(k)}{\exp(\frac{(rk)}{\hat{T}})-1}= \frac{\Delta_{3}(k)}{\exp(\frac{k}{T})-1}
\label{(27.12)}
\end{equation}
Above,
\begin{equation}
T \equiv \frac{\hat{T}}{r}
\label{(27.13)}
\end{equation}
is the measured temperature of the present CMB radiation, which is, of course, \emph{not} an equilibrium temperature. As previously discussed, \eqref{(27.13)} (with the present high value of $r$, i.e., of the order of $3000$) shows why the CMB spectrum is highly stable (metastable), because only a tiny fraction of its ``tail'' interacts with the rest of the Universe, i.e., is able to ionize a hydrogen atom. The energy $\hat{U}$ in a volume $\hat{V}$  equals
$$
\hat{U} = (\sum_{\hat{k}= \frac{2\pi}{\hat{L}}n;n \in \mathbf{Z}^{3}} \hat{k} \hat{n_{t_{0}}}(\hat{k})) \hat{V}
$$
and, correspondingly, the energy $U$ in a volume $V$ equals
\begin{eqnarray*}
U = (\sum_{k= r^{-1}\frac{2\pi}{L}n;n \in \mathbf{Z}^{3}} k n_{t}(k)) V = \\
= (\sum_{k=\frac{2\pi}{L}n;n \in \mathbf{Z}^{3}} r^{-4} \hat{k} \hat{n_{t_{0}}}(\hat{k})) r^{3}\hat{V}  
\end{eqnarray*}
and hence
\begin{equation}
U = r^{-1}\hat{U}
\label{(28.1)}
\end{equation}
By the formula (\cite{LaL}; see also \cite{LYPR} for some startling properties of this formula within the Lieb-Yngvason framework):
\begin{equation}
S_{ph}(U,V) = U^{3/4}V^{1/4}
\label{(28.2)}
\end{equation}
we obtain

\begin{proposition}
\label{prop:5.1}
The photon entropy $S_{ph}(U,V)$ is, for any fixed $V$, a constant of the cosmological evolution. The specific entropy $s(u)$, defined by
\begin{equation}
s(u) \equiv \lim_{V \to \infty} \frac{S_{ph}(U,V)}{V}  \mbox{ with } u \equiv \lim_{V \to \infty} \frac{U}{V}
\label{(28.3)}
\end{equation}
depends, however, on $T$, given by \eqref{(27.13)}, and equals
\begin{equation}
s(u) = u^{3/4} = \frac{4}{3} \sigma T^{3}
\label{(28.4)}
\end{equation}
while
\begin{equation}
u = \sigma T^{4}
\label{(28.5)}
\end{equation}
with $\sigma = \frac{\pi^{2}}{15}$. Denoting by $\rho$ the average (thermal) photon density in the thermodynamic limit, we have
\begin{equation}
\rho = \frac{2 \zeta(3)}{\pi^{2}} T^{3}
\label{(28.6)}
\end{equation}
and, therefore, from \eqref{(28.4)},
\begin{equation}
s(u) = \frac{4}{45} \pi^{2} T^{3} \approx 3.6 \rho
\label{(28.7)}
\end{equation}
\end{proposition}

\begin{proof}

The first assertion follows from \eqref{(27.11)}, \eqref{(28.1)} and \eqref{(28.2)}. The remaining assertions are standard \cite{LaL}.

\end{proof}

\begin{corollary}
\label{cor:5.1}
With model \eqref{(27.2)},
\begin{equation}
\lim_{V \to \infty} \lim_{t \to \infty} \frac{S_{ph}(U,V)}{V} \ne \lim_{t \to \infty} \lim_{V \to \infty} \frac{S_{ph}(U,V)}{V} = 0
\label{(28.8)}
\end{equation}
Above, $U$ and $V$ denote energy and volume at a certain time $t$, in agreement with the previous notation, where we suppressed the parameter $t$ for simplicity.
\end{corollary}

\begin{proof}

By the conservation law of proposition  ~\ref{prop:5.1}, the l.h.s. of \eqref{(28.8)} equals $\frac{4}{3} \sigma \hat{T}^{3} \ne 0$, while, by \eqref{(27.13)} and \eqref{(28.4)}, together with \eqref{(27.2)} and \eqref{(27.10)}, the r.h.s. of \eqref{(28.8)} indeed equals zero. This result only depends, of course, on \eqref{(27.3)} and generalizes to any expanding Universe.

\end{proof}

Corollary ~\ref{cor:5.1} shows that the entropy of the CMB photons satisfies Second law (Cl) (trivially, being conserved), but not the Modified Second Law (Cl). this is due to the (general) noncommutativity of the thermodynamic limt with the limit of large times $t \to \infty$, previously mentioned at the end of section 4; it is here shown explicitly for this model. On the other hand, by \eqref{(28.1)}, First Law (Cl) is \emph{not} satisfied, due to the loss of energy of the CMB photons during the expansion of the Universe, as a consequence of the cosmological red-shift \eqref{(27.10)}.

According to \eqref{(27.12)}, we live immersed in a (canonical) state of radiation whose average (thermal) photon density $\rho = 20 T_{p}^{3}$, which amounts to four hundred photons per cubic centimeter - a colossal number, quoting Harrison (\cite{Har}, p.274): ``in one second, $10^{15}$ CMB photons will reach the surface of your hand, at least a factor of $10^{5}$ of all the photons which have been radiated by the stars, and, by \eqref{(28.7)} the largest part of the specific entropy of the Universe is already in the background radiation and will be hardly affected by the future behavior of the stars''. This specific entropy violates, however, Modified Second Law (Cl), and this violation is ultimately for the same reason why the First Law (Cl) is violated, namely, the cosmological red shift \eqref{(27.10)}. It seems therefore of great importance to understand the failure of both the first and (modified) second law of thermodynamics for the Universe. As a preliminary, one may pose the question: \emph{what} replaces the First Law (Cl) in the case of the Universe?

A gravitational energy-momentum pseudotensor has been proposed by Landau and Lifschitz in their text in classical field theory (\cite{LaLTC}, Chap. 100, pp. 379-381), and is called the Landau-Lifschitz pseudo-tensor $t_{LL}^{\mu\nu}$. When added to the energy-momentum of matter (which includes photons and neutrinos) $T^{\mu\nu}$ is such that its total divergence vanishes, i.e.,
\begin{equation}
((-g)(T^{\mu\nu}+t_{LL}^{\mu\nu})_{,\mu} = 0
\label{(29.1)}
\end{equation}
where $g$ denotes the determinant of $g^{\mu\nu}$.
Thirring (with Wallner) derives $t_{LL}^{\mu\nu}$ from the Landau-Lifschitz 3-form in his article in Rev. Bras. Fis. (\cite{ThirrRBF}, Appendix C). This yields a formula for $t_{LL}^{\mu\nu}$ in terms of the metric tensor $g^{\mu\nu}$ and its derivatives. This article was one of the first to explore the use of E. Cartan's exterior forms in Einstein's gravitation theory, following the pioneer work of unification of electromagnetism and gravitation by M. Schoenberg \cite{SchRBF}. Cartan's formalism, besides its elegance, is also very economical, allowing to write down a compact explicit formula for $t_{LL}^{\mu\nu}$ ((C.10) of Appendix C of \cite{ThirrRBF}), which neither I, nor the authors of \cite{ThirrRBF}), have seen written down elsewhere. From this formula, it is readily seen that $t_{LL}^{\mu\nu}$ contains only first derivatives of the metric tensor $g^{\mu\nu}$, and these may be made to vanish at any chosen point, upon choice of a frame which is locally inertial at this point. This follows from the mass-energy equivalence principle and, as a consequence, at any chosen point, $t_{LL}^{\mu\nu} = 0$, demonstrating why $t_{LL}^{\mu\nu}$ cannot be a tensor, as well as the important fact that the gravitational energy-momentum is delocalized. In spite of this, from \eqref{(29.1)} and Stokes' theorem, we obtain (\cite{Thirr3rd}, Cor. 7.3.35, n.1):

\emph{First Law (Cl)} The total energy and momentum are conserved, as long as energy and momentum fall off sufficiently fast at infinity on the submanifold $t= \mbox{ const. }$.

Under suitable conditions, it is also expected that the energy and momentum per unit volume in the thermodynamic limit exist and are conserved: this corresponds to versions of Modified First Law (Cl), which may be expected to be of greater relevance from a physical standpoint, since only the specific energy and momentum are accessible to experiment. It should be emphasized that, in spite of the loss of energy of (a part of) matter due to \eqref{(28.1)}, the pseudotensor $t_{LL}^{\mu\nu}$ does lead to a conservation law due to the \emph{cancellation} of the Einstein tensor with the matter stress-energy tensor $T^{\mu\nu}$ by the Einstein field equations, see (\cite{LaLTC}, p. 381). 

The above-mentioned cancellation does *not* occur, however, within the previously made approximation that the energy-momentum tensor of the photons is independent of the metric!. For this reason the validity of First Law (Cl) remains a challenging open problem.

Inclusion of the gravitational field is, in principle, also able to show that Modified second law (Cl) holds for the \emph{total} specific entropy. This is, however, an even more challenging open problem than the one associated to the first law: indeed, \emph{nothing} is rigorously known about the specific entropy of the gravitational field!

As a final remark, a term containing the cosmological constant
$$
-\frac{c^{4}}{8 \pi G} \Lambda g^{\mu\nu}
$$
should be added to $t_{LL}^{\mu\nu}$. Above, $\Lambda$ denotes the (Einstein) cosmological constant as before, and $G$ the gravitational constant. In this connection, it may be mentioned that \emph{scale invariance} of the macroscopic empty space, which intervenes through the cosmological constant, leads to a consistent theoretical framework, where neither dark energy or dark matter are needed, as shown in the beautiful work of Maeder (see \cite{Maeder} and references given there). These brief remarks reflect, however, my personal view only; for a deeper (and comprehensive) discussion of the issues of dark matter, dark energy and the cosmological constant, see \cite{Straumann}.   

\section{Conclusion}

Vaclav Havel has stated \cite{Havel} that ``the problem of modern man is not that he understands less and less the meaning of life, but that this fact has almost ceased to bother him''. I shall not insist on the obvious parallel of this quote to Sommerfeld's, but hope that the present manuscript may stimulate some readers in ``bothering further'' about the deep open problems associated to the fundamental laws of physics, in particular those of thermodynamics, some of which have been discussed in sections 3, 4 and 5.

\section{Acknowledgement} I should like to thank Elliott Lieb for his remarks and suggestions regarding this manuscript, Manfred Requardt for discussions and remarks concerning the Landau-Lifschitz pseudotensor, Pedro L. Ribeiro for his perceptive criticism of the CMB model and the observation concerning Klein's paradox, and Geoffrey Sewell for his observation on the time-arrow problem. I am also grateful to the referee for several constructive remarks, in particular for drawing my attention to the very relevant section 3.8 of Peierls' book.

\end{document}